\documentclass[runningheads,10pt]{llncs}
\usepackage{graphicx}

\usepackage{latexsym}
\usepackage{amsmath}
\usepackage{mathrsfs}
\usepackage{amssymb}
\usepackage{epsfig}
\usepackage{color}
\usepackage{todonotes}
\usepackage[ruled,vlined,linesnumbered]{algorithm2e}
\usepackage{algpseudocode}
\usepackage{xspace}
\usepackage{float}
\usepackage{mathtools}
\usepackage{tabularx}
\usepackage{relsize}
\usepackage[pagewise]{lineno}
\usepackage{enumitem}

\newtheorem{mytheorem}{Theorem}[section]
\numberwithin{mytheorem}{section}
\newtheorem{mylemma}{Lemma}[section]
\numberwithin{mylemma}{section}

\numberwithin{clm}{mytheorem}

\numberwithin{obs}{mytheorem}

\let\oldnl\nl
\newcommand{\nonl}{\renewcommand{\nl}{\let\nl\oldnl}}

\newcommand{\OO}{{\mathcal O}}
\newcommand{\npc}{\textsf{NP}-complete\xspace}
\newcommand{\nph}{\textsf{NP}-hard\xspace}
\newcommand{\no}{{\emph no}\xspace}
\newcommand{\yes}{{\emph yes}\xspace}

\newtheorem{innercustomgeneric}{\customgenericname}
\providecommand{\customgenericname}{}

\usepackage{todonotes}
\usepackage{complexity}

\newtheorem{preprocessing rule}{Preprocessing Rule}
\newtheorem{reduction rule}{Reduction Rule}
\newtheorem{branching rule}{Branching Rule}

\newcommand{\prob}[3]{
  \vspace{-2mm}
\noindent\fbox{
  \begin{minipage}{0.96\textwidth}
  \begin{tabular*}{\textwidth}{@{\extracolsep{\fill}}lr} \textsc{#1}  
\\ \end{tabular*}
  {\bf{Input:}} #2  \\
  {\bf{Question:}} #3
  \end{minipage}
  }
  \vspace{5mm}
}

\begin{document}

\title{Role Coloring Bipartite Graphs}
%
%
\author{Sukanya Pandey\inst{1}\and
Vibha Sahlot\inst{2}}
\authorrunning{S. Pandey et al.}
%
\institute{Utrecht University, The Netherlands \\
\email{s.pandey1@uu.nl} \and
Charles University in Prague, Czechia\\
\email{sahlotvibha@gmail.com}
}
\maketitle              

\begin{abstract}
A $k$-role coloring $\alpha$ of a graph $G$ is an assignment of $k$ colors to the vertices of $G$ such that every color is used at least once and if any two vertices are assigned the same color, then their neighborhood are assigned the same set of colors. By definition, every graph on $n$ vertices admits an $n$-role coloring.

While for every graph on $n$ vertices, it is trivial to decide if it admits a $1$-role coloring, determining whether a graph admits a $k$-role coloring is a notoriously hard problem for $k \geq 2$. In fact, it is known that $k$-{\sc Role coloring} is \npc for $k \geq 2$ on general graphs.

There has been extensive research on the complexity of $k$-{\sc role coloring} on various hereditary graph classes. Furthering this direction of research, we show that $k$-{\sc Role coloring} is \npc on bipartite graphs for $k \geq 3$ (while it is trivial for $k=2$). 
We complement the hardness result by characterizing $3$-role colorable bipartite chain graphs, leading to a polynomial time algorithm for $3$-{\sc Role coloring} for this class of graphs. We further show that  $2$-{\sc Role coloring} is \npc for graphs that are $d$ vertices or edges away from the class of bipartite graphs, even when $d=1$.
\end{abstract}

\section{Introduction}

Given a graph $G$, how would you assign $k$ colors to its vertices such that every color is assigned to some vertex and if any two vertices get the same color, then their neighborhood are assigned the same set of colors? Is such an assignment possible for any input graph? Known as $k$-\textsc{Role coloring}, this problem was motivated by applications in sociology. It was introduced as a graph coloring problem by Borgatti and Everett~\cite{rcol} in 1991. The underlying principle was that in a social network, individuals play the same role if they relate in the same way to other individuals playing counterpart roles.

Formally, for any input graph $G$, a $k$-role coloring is an assignment of exactly $k$ colors to its vertices such that if any two vertices get the same color, then the set of colors assigned to their neighborhood is also 
the same. That is, $k$-role coloring is a surjective map
$\alpha: V(G) \rightarrow \{1, 2, ..., k\}$ such that for all $u, v \in V(G)$, if $\alpha(u) = \alpha(v)$ then $\alpha(N(u)) = \alpha(N(v))$. For a function $\alpha$ on domain $U$ and a set $A \subseteq U$, we abuse the notation $\alpha(A)$ to denote the set $\{\alpha(a)|a \in A\}$.

Corresponding to a $k$-role assignment $\alpha$, the role graph $R$ is defined as the graph with $V(R)=  \{1, 2, ..., k\}$ and $E(R)= \{ (\alpha(u), \alpha(v)) |  (u,v) \in E(G)\}$. Since each color is assigned to some vertex of $G$, it is easy to see that when $G$ is connected, the role graph $R$ is also connected and $|V(R)| \leq |V(G)|$. 
Also, for all $v \in V(G)$, $deg_G(v) \geq deg_R(\alpha(v))$. 

We define the problem $k$-\textsc{Role coloring} as follows,\\

 \prob{$k$-{\sc Role coloring}}{An undirected graph $G$ and an integer $k$.}{Does there exist a surjective function $\alpha: V(G) \rightarrow \{1 ,2 ,..., k\}$ satisfying: if $\alpha(u) = \alpha(v)$, then $\alpha(N(u)) =\alpha(N(v))$ for all  $u, v \in V(G)$?}
 
 We can also define the problem in terms of role graphs. By $N_G(u)$ and $N_R(v)$, we refer to the vertices in the neighborhood of the vertices $u \in V(G)$ and $v \in V(R)$, respectively. \\

 \prob{$R$-{\sc Role coloring}}{Undirected graphs $G$ and $R$.}{Does there exist a surjective function $\alpha: V(G) \rightarrow V(R)$ satisfying: $\alpha(N_G(u)) = N_R(\alpha(u))$ for all  $u \in V(G)$?}
 
This problem is equivalent to deciding if there exists a {\em locally surjective} homomorphism between the graphs $G$ and $R$~\cite{Rrcol}. 
Notice that given a role graph $R$ and a role coloring $\alpha$ of $G$, we can find in polynomial time if $\alpha$ agrees with $R$, that is whether $G$ is $R$-role colorable with respect to $\alpha$. Also, given a role coloring $\alpha$ of $G$, we can find whether $\alpha$ is a valid role coloring in polynomial time. Further, we can also find the role graph $R$ corresponding to $\alpha$.

\smallskip

\paragraph{Previous Work:} 
Fiala et al.~\cite{Rrcol} showed that $R$-{\sc Role coloring} can be solved in $n^{\mathcal{O}(1)}$ time when
\begin{itemize}
	\item $R$ has no edges, or
	\item it has a component isomorphic to a single loop-incident vertex, or 
	\item it is simple and bipartite and has at least one component isomorphic to a~$K_2$. 	
\end{itemize}

For all other cases, they proved that $R$-{\sc Role coloring} is \npc. They used this result to show that $k$-{\sc Role coloring} is \npc for all $k \geq 3$. Prior to their work, it was shown by Roberts and Sheng~\cite{2rcol} that $k$-{\sc Role coloring} is \npc for $k=2$. Hence, the following dichotomy exists for $k$-{\sc Role coloring}: Given any $n$ vertex graph $G$, $k$-{\sc Role coloring} is polynomial time solvable when $k= 1$ or $n$ and \npc otherwise.

Various studies have been conducted to find the complexity of computing role-coloring on hereditary graph classes. Sheng~\cite{2chordal} gave a greedy algorithm to compute a $2$-role coloring of chordal graphs in polynomial time. Later,Heggernes et al.~\cite{chordal} showed that $k$-{\sc role coloring} is \npc on chordal graphs when $k \geq 3$. Following these results, Dourado~\cite{split} showed that on split graphs, a subclass of chordal graphs, $3$-{\sc role coloring} can be decided in linear time, whereas for $k \geq 4$, $k$-{\sc role coloring} is \npc on split graphs. The problem remains \npc on planar graphs for $k \geq 2$~\cite{planar}. The only known classes of graphs for which $k$-{\sc role coloring} can be decided in polynomial time for any fixed $k$, are Cographs~\cite{planar} and Proper Interval graphs~\cite{pig}.

Computational complexity and algorithms for many related variants of graph colouring have been studied in the literature. One of the most common variants is {\sc Perfect coloring}. In this problem, the aim is to find the minimum number of colors such that the endpoints of each edge in the graph are colored differently. It is \npc to decide if a given graph admits a perfect coloring with $k$ colors except for the cases $k \in \{0,1,2\}$~\cite{properColoring}.  This coloring problem remains \npc for $k=3$ even on $4$-regular planar graphs~\cite{DAILEY1980289}. A closely related variant is {\sc Ecological Coloring}. A  coloring  is ecological  if  any two vertices that are  surrounded  by  the vertex set with the same set of colors are colored the same. In this problem, our aim is to decide whether there exists an ecological coloring for a given integer $k$ in a graph. The computational complexity of this problem has also been studied for different graph classes ~\cite{BORGATTI199443,BORGATTI1992287,CrescenziIGRV08}. Another related problem is {\sc Coupon-coloring}. In this, we aim to find the maximum number of colors such that the open neighborhood of each vertex contains all the colors~\cite{CHEN201594}. This problem has been investigated for several graph families~\cite{CHEN201594,CouponColoring7,CouponColoring13}. 
\\
\paragraph{Our Contribution:}
Continuing the line of research on complexity of {\sc role coloring} on hereditary graph classes, we explore this problem when restricted to the class of bipartite graphs. Due to the work of Fiala et al.~\cite{Rrcol}, there exists a complexity dichotomy for $R$-{\sc Role coloring} on the class of bipartite graphs. We build on this result and focus on the more general problem of $k$-{\sc Role coloring} on bipartite graphs. We show that $k$-{\sc Role coloring} is \npc on the class of bipartite graphs for all $k \geq 3$ (while it is trivial when $k \leq 2$).

Note that the \textsf{NP}-completeness of $k$-{\sc Role coloring} does not follow from the \textsf{NP}-completeness of $R$-{\sc Role coloring}, where $R$ is a graph on $k$ vertices satisfying the characterization in~\cite{Rrcol}. To exemplify this, let us consider the two variants of the problem with respect to the class of split graphs. Let $R_1$ be a graph isomorphic to a $K_2$ with a self loop on exactly one end point, and $R_2$ be isomorphic to a $K_2$ with self loops incident on both the end points. It has been shown in~\cite{chordal} that $R_2$-{\sc Role coloring} is \npc on split graphs. However, $2$-{\sc Role coloring} is trivially solvable in polynomial time on split graphs with the corresponding role graph being isomorphic to $R_1$. 

Complementing our hardness result, we give a polynomial time algorithm for $3$-{\sc Role coloring} on bipartite chain graphs. Bipartite chain graphs are bipartite graphs such that for the vertices in each partition, their neighborhood can be ordered linearly with respect to inclusion. Equivalently, they are bipartite graphs that exclude $2K_2$ as an induced subgraph. (A $2K_2$ is a set of two disjoint edges). 

Finally, we investigate the problem $2$-{\sc Role coloring} on graphs that are $d$ vertices (or edges) away from bipartite graphs. We show that the problem is \npc even when $d=1$. Henceforth, we shall refer to such a graph $G$ as an {\em ``almost bipartite"} graph. Parameterization by distance from a graph class was first studied for the vertex coloring problem in~\cite{almost}. We draw inspiration from that and explore it with respect to {\sc $k$-Role coloring}.
\\
\paragraph{Organization of the paper:} The paper is organised as follows. We define the terminology and notations used in this paper in Section \ref{sec:prelim}. 

Section \ref{sec:BipNPComp} gives the proof for \textsf{NP}-completeness of $k$-{\sc Role coloring} on bipartite graphs, for fixed $k\geq 3$. The proof is divided into three parts. We show separately the \textsf{NP}-completeness for each of the cases when $k\in \{3,  4\}$ and for the case when $k \geq 5$. In Section \ref{sec:poly}, we  characterize bipartite chain graphs that are $3$-role colorable, and in Section \ref{sec:para}, we show that $2$-{\sc role coloring} is \textsf{NP}-complete on ``almost bipartite graphs".
 
 \section{Preliminaries} \label{sec:prelim} 
 In this section, we state the graph theoretic terminology and notation used in this paper. The set of consecutive integers from $1$ to $n$ is denoted by $[n]$. Let $f:A \rightarrow B$ be a function. Then, for any non-empty set $A' \subseteq A$, by $f(A')$, we denote the set $\{f(a)| a \in A'\}$. The vertex set and the edge set of a graph $G$ are denoted by $V(G)$ and $E(G)$, respectively (or simply $V$ and $E$ when the underlying graph $G$ is clear from the context). By $|G|$, we denote the order of $G$, that is the number of its vertices $|V(G)|$. An edge between vertices $u$ and $v$ is denoted as $(u,v)$.  

For a vertex $v \in V(G)$, its {\em neighborhood} $N_G(v)$ is the set of all vertices adjacent to it and its {\em closed neighborhood} $N_G[v]$ is the set $N_G(v) \cup \{v\}$. This notation is extended to subsets of vertices as $N_G[S]=\bigcup_{v \in S}N_G[v]$ and $N_G(S)=N_G[S] \setminus S$ where $S \subseteq V(G)$. The {\em degree} of a vertex $v \in V(G)$, denoted by $deg_G(v)$, is the size of $N_G(v)$. A {\em pendant vertex} is a vertex $v \in V(G)$ with $deg_G(v) = 1$.  
A graph $G'$ is a \emph{subgraph} of $G$ if $V(G')\subseteq V(G)$ and $E(G')\subseteq E(G)$. A graph $G'$ is an \emph{induced subgraph} of $G$ if for all $x,y \in V(G')$ such that $(x,y)\in E(G)$, then $(x,y)\in E(G')$.  A {\em path} in a simple graph is a sequence of distinct vertices with an edge between every pair of  consecutive vertices. The {\em length of a path} is the number of vertices in a path. We denote a path on $n$ vertices by $P_n$. A {\em pendant path} is an induced path attached to a vertex. For further details on graphs, refer to~\cite{diestel2000graph}.


A \emph{hypergraph} is a pair $H=(Q_H, S_H)$ where $Q_H$ is a finite set of vertices and $S_H$ is a collection of non-empty subsets of $Q_H$ called \emph{hyper-edges}. We define the problem {\sc Hypergraph $k$-coloring} as follows:\\
\prob{Hypergraph $k$-coloring}{A connected hypergraph $H = (Q_H,S_H)$.}{Does there exist a coloring which assigns $k$ colors to vertices in $Q_H$ such that no hyperedge $S \in S_H$ is monochromatic?}\\

A hypergraph is called $k$-colorable if its vertices can be assigned $k$ colors such that none of its edges is monochromatic. A hypergraph is $q$-uniform if each of its edges is a subset of its vertices of size $q$. Lov{\'{a}}sz~\cite{Hypergraph-2-coloring} showed that it is \nph to test $2$-colorability of $3$-uniform hypergraphs. The following theorem by Khot implies that \textsc{Hypergraph-k-coloring} is \nph on $3$-regular hypergraphs, for any fixed $k>1$.

\begin{theorem} \cite[Theorem~1.1]{Hypergraph-k-coloring} \label{kHypergraphColoring}
For every constant $\delta> 0$, it is \nph to distinguish whether an $n$-vertex $3$-uniform hypergraph is $3$-colorable or it contains no independent set of size $\delta n$. In particular, it is \nph to color
$3$-colorable $3$-uniform hypergraphs with constantly many colors.
\end{theorem}

Following are some observations about the role graph that we use throughout the paper without explicitly referring to them.


\smallskip
\noindent\textbf{Observation 2.1} (\cite{Rrcol}) If $G$ has an $R$-role coloring $\alpha$, then $deg_G(u) \geq deg_R(\alpha(u))$ for all vertices $u \in V(G)$.

\smallskip
\noindent\textbf{Observation 2.2} (\cite{Rrcol})
If $G$ is a connected graph with a valid $k$-role coloring $\alpha$, then the role graph corresponding to $\alpha$ must be connected.\\
\noindent{\em Proof.}
Since $G$ is connected, there is a path between every pair of vertices $u, v$ $\in V(G)$. Let $P= u{x_1}{x_2}...{x_{l-2}}v$ be a path of length $l$ in $G$. Let $R$ be the role-graph corresponding to $\alpha$. By definition, $\alpha(N_{G}(u))= N_{R}(\alpha(u))$. That is, $\alpha(x_1) \in N_R(\alpha(u))$. By induction on $l$, it can be shown that there exists a path between $\alpha(u)$ and $\alpha(v)$ in $R$. Since the choice of $u$ and $v$ was arbitrary, the claim is true for all pairs of vertices in $R$. Hence, $R$ is connected.\qed

\section{\textsf{NP}-completeness of $k$-{\sc Role coloring} on bipartite graphs} \label{sec:BipNPComp}

In this section, we prove that $k$-{\sc Role coloring} is \npc on the class of bipartite graphs, when $k \geq 3$.

\begin{mytheorem}\label{thm:1} 
$k$-{\sc Role coloring} is \npc on bipartite graphs, for any fixed $k \geq 3$.
\end{mytheorem}

Given a $k$-role coloring of $G$, we can verify if it is valid in $n^{\OO(1)}$ time by checking for each vertex $v$ in a color class, if its neighborhood is assigned the same set of colors. Hence, $k$-{\sc role coloring} is in \textsf{NP}.
We prove the \textsf{NP}-hardness in four parts. We separately deal with the cases when $k=3$, $k=4$, and $k \geq 5$. We show a reduction from \textsc{Hypergraph-$2$-coloring} for $k=3$ and $k \geq 5$, and from \textsc{Hypergraph-$3$-coloring} for $k=4$.

For each case, the hypergraph is denoted by $H = (Q_H, S_H)$, where $Q_H$ is a finite set of vertices and $S_H$ is a collection of non-empty subsets of $Q_H$ called hyperedges such that $|s|=3$ for each $s \in S_H$. The corresponding \emph{canonical incidence graph} of the hypergraph $H=(Q_H, S_H)$ is a bipartite graph denoted as $G=(Q, S, E)$ with $Q$ and $S$ 
as the two parts of the  bipartition, and for all $ q \in Q$ and $s \in S$, $(q, s) \in E(G)$ if and only if the vertex corresponding to $ q  \in Q_H$ belongs to the hyperedge $s \in S_H$. 

\begin{mylemma}
$3$-{\sc Role coloring} is \npc on connected bipartite graphs.
\end{mylemma}\label{lem:equal3} 

\begin{proof}
Consider the incidence graph $G=(Q, S, E)$ of a hypergraph $H=(Q_H,S_H)$.
To each vertex $ q  \in  Q$, we add a path on two vertices, such that the vertex adjacent to $q$ is labeled $b_q$ and the other vertex adjacent to $b_q$ is $a_q$. For each $q \in Q$, the sub-graph induced on $\{q, b_q, a_q\}$ is a ``pendant path"  and has length $3$. The new graph thus obtained is denoted by $G'$ (refer to Figure \ref{thm1pic}). 

\begin{figure}[!ht]
\begin{center}
\includegraphics[width=0.5\textwidth]{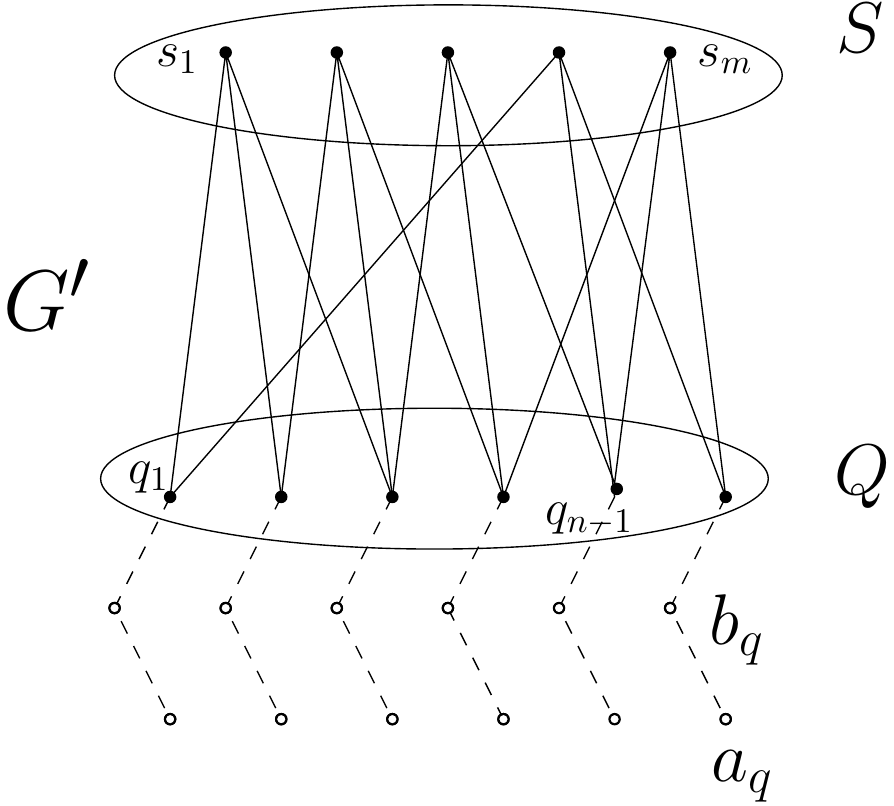}
\caption{$G'$ is a connected bipartite graph constructed from the incidence graph $G$ of a $3$-regular hypergraph $H$. To each vertex $q\in Q$, we add a $P_2$. The vertex adjacent to $q$ is denoted by $a_q$ and the other one by $b_q$.} 
\label{thm1pic}
\end{center}
\end{figure}

We claim that $H$ is a \yes instance of {\sc Hypergraph $2$-coloring} if and only if  $G'$ is a \yes instance of $3$-{\sc Role coloring}. First, assume that $H$ is a \yes instance and $\beta: Q_H \rightarrow \{1, 2\}$ is a 2-coloring of the vertices in $Q_H$.
Then, the role coloring $\alpha: V(G') \rightarrow \{1, 2, 3\}$ can be defined as: for all $q \in Q$, $\alpha(q)=\beta(q)$;  $\alpha(b_q)= 3$; for all $s\in S$ $\alpha(s)= 3$, and $\alpha(a_q) = j, j \in \{1,2\}\setminus \alpha(q)$. The coloring is consistent as each vertex colored $1$ is adjacent to some vertex colored $3$. Similarly, every vertex colored $2$ is adjacent to some vertex colored $3$. Finally, since each hyperedge in $S_H$ is bichromatic, the corresponding vertices in $S$ are each adjacent to some vertex colored $1$ and some vertex colored $2$. Also, each vertex $b_q$ is adjacent to one vertex colored $1$ and one colored $2$. $G'$ is, therefore, a \yes instance of $3$-{\sc role coloring}.

Conversely, suppose that $G'$ is a \yes instance of $3$-{\sc role coloring} with the $3$-role coloring assignment $\alpha'$ and the corresponding role graph $R'$. 
Without loss of generality, assume that for some $q \in Q$, $\alpha'(a_q) = 1$. Then, $\alpha'(b_q) \neq 1$, or else the role graph would be disconnected (with one component isomorphic to a single vertex with a loop). This is a contradiction as $G'$ is connected. So let $\alpha'(b_q)=2$. Now, $\alpha'(q) \neq 1$ and $\alpha'(q) \neq 2$, otherwise role graph would again be disconnected. So, let $\alpha'(q)=3$. 
We prove that $\alpha'(S) = \{2\}$. 
The neighbor of $q$, say $s \in S$ could either be colored $2$ or $3$. Suppose that $\alpha'(s)= 3$. This forces one of the neighbors of $s$ to be colored $2$. Since, no hyperedge in $S_H$ is singleton, let there be a neighbor $q'$ of $s$ such that $\alpha'(q') = 2$. Consider the vertex $b_{q'}$ adjacent to $q'$. $\alpha'(b_{q'}) \neq 1$ as that would force $\alpha'(a_{q'}) = 2$. This is not possible because $a_{q'}$ has no neighbor colored $3$. Similarly, $\alpha'( b_q') \neq 3$. Therefore, in a valid role coloring of $G'$, $\alpha'(s) \neq 3$, for any $s \in S$. Also $\alpha'(s) \neq 1$ for $s \in S$ as $3$ is not adjacent to $1$ in the role graph. Hence, $\alpha'(s) = 2$ for all $s \in S$. We define the coloring of the hypergraph $\beta: Q_H \longrightarrow \{1, 3\}$, as follows as $\beta(q)= \alpha'(q)$. This is a valid $2$-coloring of $H$, as every $s \in S$ is adjacent to at least one vertex colored $1$ and at least one vertex colored $3$. Hence, $H$ is a \yes instance of {\sc Hypergraph-$2$-coloring}.\qed

\end{proof} 


\begin{mylemma}
$4$-{\sc Role coloring} is \npc on bipartite graphs.
\end{mylemma}\label{lem:2}

\begin{proof}
Let $H$ be an instance of {\sc Hypergraph $3$-coloring} such that $H$ is $3$-regular. The incidence graph $G=(Q, S, E)$ is a connected bipartite graph, with the bipartition $(Q, S)$. To construct the graph $G'$ from $G$, we add a pendant vertex $p^s$ to every vertex $s \in S$ (refer to Figure \ref{thm2pic}). Let $P^S=\{p^s\mid s\in S\}$.

\begin{figure}[!ht]
\begin{center}
\includegraphics[width=0.5\textwidth]{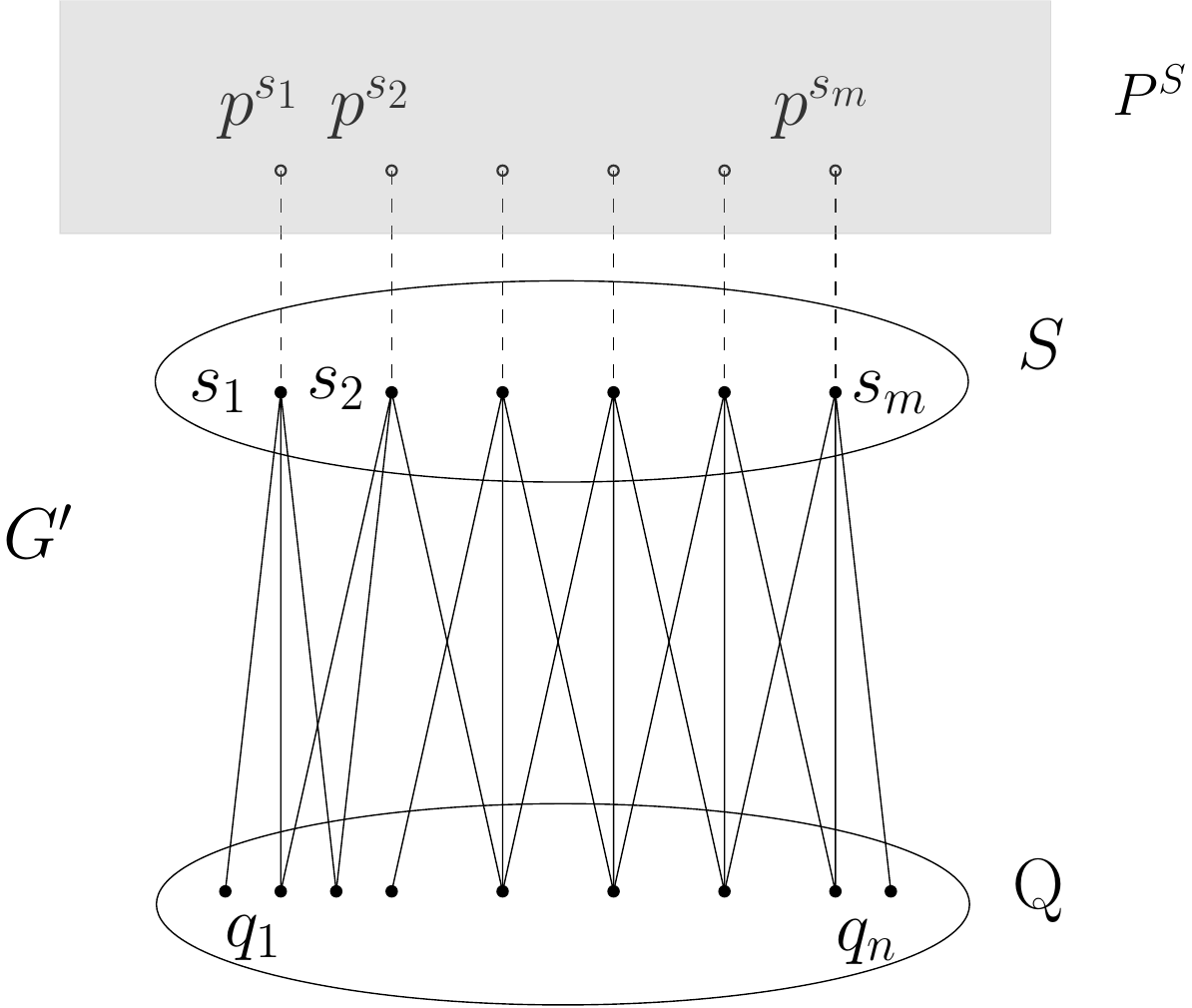}
\caption{$G'$ shown in the figure is a connected bipartite graph with the bipartition $(Q \cup P^S, S)$. It is constructed from the incidence graph $G$ of the hypergraph $H$ by adding a pendant vertex to each vertex in $S$.} 
\label{thm2pic}
\end{center}
\end{figure}

We now claim that $H = (Q_H, S_H)$ admits a $3$-coloring if and only if $G'$ admits a $4$-role coloring. Suppose that $H$ has a $3$-coloring, $\beta$ : $Q_H \rightarrow \{1, 2, 3\}$. 
Then the role coloring for $G'$ is defined as: for each $q \in Q$, $\alpha(q)$=$\beta(q)$; for all $s \in S$, $\alpha(s)= 4$ and for each $s$, if $|\alpha(N_{Q}(s))| = 2$, then $\alpha(p^{s})= x \in \{1, 2, 3\}$ such that $x \notin \alpha(N_{Q}(s))$. Otherwise, arbitrarily assign any of the three colors to $p^s$.  Corresponding to $\alpha$, the role graph $R$ is defined as $V(R)= \{1, 2, 3, 4\}$ and $E(R)= \{(1,4), (2,4), (3,4)\}$. As $H$ is a \yes instance of {\sc Hypergraph $3$-coloring}, every vertex $s \in S$ colored $4$ is adjacent to vertices colored $1$, $2$ and $3$ in $Q \cup P^s$. Thus, $\alpha$ is consistent with $R$.

Conversely, suppose that $G'$ has a valid $4$-role coloring $\alpha': V(G) \rightarrow \{1, 2, 3, 4\}$, and the corresponding role graph is $R'$.

\paragraph{ Claim~1.}
{\em In any valid $4$-role coloring of $G'$, $\alpha'(S) \cap \alpha'(Q) = \emptyset$.}

\noindent\textit{Proof.} Towards a contradiction, let us suppose that there exists an $s\in S$ such that $\alpha'(s)=\alpha'(q) = x$, for some $q \in Q$ and $x\in \{1,2,3,4\}$. Then $\alpha'(p^s)\neq x$ or else $R'$ would be disconnected with a component containing $x$ with a loop. By Observation~2.2, this is not possible. Therefore, let $\alpha'(p^s)=y$, such that $y \in \{1,2,3,4\}\setminus\{x\}$. For $\alpha'$ to be valid, $q$ must have a neighbour $s'\in S$ such that $\alpha'(s')=y$. Consider $p^{s'}$. Since $N_{R'}(y)=\{x\}$, $\alpha'(p^{s'})$ must be equal to $x$. However, by Observation~1.1, this is not possible because $deg_{R'}(x)>1$. \hfill $\qed$\\

\paragraph{ Claim~2.}
{\em When $\vert \alpha'(Q)\vert = 1$, $G'$ does not admit any valid $4$-role coloring.}

\noindent\textit{Proof.} Towards a contradiction, assume that $Q$ is a  monochromatic set. Without loss of generality let $\alpha'(q)=1$, for each $q \in Q$.

By claim~1, no vertex in $S$ is assigned the color $1$. We show that none of the pendant vertices in $P^s$ gets the color $1$. Suppose that there exists a vertex $p^{s}$ of $P^s$ colored $1$. Suppose that $\alpha'(s)$ = $2$. Since the vertices colored $1$ are only adjacent to vertices colored $2$, $S$ must be monochromatic with the color $2$, as by our assumption $Q$ is monochromatic with the color $1$. Therefore, $P^s$ must be monochromatic with the color $1$ assigned to its vertices as the vertex $s$ colored $2$ is only adjacent to vertices colored $1$, a contradiction to $\alpha'$ being a $4$-role coloring of $G'$. Therefore, it follows that $1 \notin  \alpha'(P^s \cup S)$. 
 
We now prove that $P^s$ must be monochromatic.
Towards a contradiction, suppose that there exist vertices $p^{s}$ and $p^{s'}$ colored $2$ and $3$, respectively. Then $\alpha'(s) \neq 2$ and $\alpha'(s') \neq 3$, or else every vertex colored $2$ or $3$ could only have neighbors of the same color. As every $s \in S$ has a neighbor in $Q$ colored $1$, we know that that is not possible. This implies that $s \in S$ is adjacent to vertices colored $1$ and $2$, whereas the $s'\in S$ has neighbors colored $1$ and $3$. Suppose that we color $s$ with $4$, then there is no color that can be assigned to $s'$. This is a contradiction to $G'$ admitting a valid $4$-role coloring. Hence, $P^s$ must be monochromatic.

Let $\alpha'(P^s)$=$\{2\}$. That forces $S$ to be monochromatic with a third color, say $3$, which again contradicts the assumption that $G'$ admits a valid $4$-role coloring. Therefore, when $Q$ is monochromatic, $G'$ does not admit any $4$-role coloring. \hfill $\qed$\\

\paragraph{ Claim~3.}
{\em When $\vert\alpha(Q)\vert= 2$, $G'$ has a valid $4$-role coloring.}

\noindent\textit{Proof.} Let the colors assigned to vertices in $Q$ by $\alpha'$ be $1$ and $2$.

By claim~1, $\alpha'(S) \cap \{1,2\}=\emptyset$. Next, we show that the pendants in $\alpha'(P^s) \cap \alpha'(Q)= \emptyset$. Suppose that there exists a pendant vertex $p^s$ colored $1$. Let its neighbor $s \in S$ be colored $3$. Also, $3$ must be adjacent to a different color, say $2$,  otherwise the role graph would be disconnected with a component containing just $1$ and $3$. This means that no pendant receives the color $3$ as the degree of $3$ in $R'$ is at least equal to 2. If there exists a pendant colored $2$, then its neighbor must be colored $3$ as well. So, $N_{R'}(1)= N_{R'}(2) = \{3\}$ and $N_{R'}(3)= \{1, 2\}$. All the vertices in $Q$ are adjacent only to vertices in $S$. Hence, $S$ is monochromatic with the role $3$. This is a contradiction to $\alpha'$ being a $4$-role coloring. This implies that $1 \notin \alpha(P^s)$ and $ 2 \notin \alpha'(P^s)$. This proves our claim that $\alpha(Q) \cap \alpha'(P^s)$ = $\emptyset$. 

Thus, the only possible role coloring in this case is  $\alpha'(Q) \in \{1,2\}$, $\alpha'(S)= \{3\}$, and $\alpha'(P^s)= \{4\}$. So, every vertex in $S$ is forced to be adjacent to at least one vertex having color $1$ and at least one vertex having color $2$. 

The corresponding coloring for the hypergraph $H$ is as follows: for some $q' \in Q_H, \beta(q')= 3$ and $\forall q \in Q_H \setminus \{q'\}, \beta(q) = \alpha(q)$. Clearly, no hyperedge is monochromatic and $|\beta(Q_H)| = 3$. \hfill $\qed$\\

\paragraph{ Claim~4.} {\em When $\vert \alpha'(Q)\vert =3, G'$ admits a valid $4$-role coloring.}

\noindent\textit{Proof.} Suppose that the colors assigned to $Q$ by $\alpha'$ are $\{1,2,3\}$. 

By Claim~1, $\alpha'(S) \cap \alpha'(Q) = \emptyset$. Therefore, $\alpha'(S)=\{4\}$ and $R'=(\{1,2,3,4\},$ $ \{(1,4), (2,4), (3,4)\})$. We propose the following coloring for the hypergraph $H$. For every $q \in Q_H$, $\beta(q) = \alpha'(q)$ such that $q \in Q$ is the corresponding vertex in $G'$. \hfill $\qed$\\ 

\paragraph{ Claim~5:}{\em When $\vert \alpha'(Q)\vert = 4$, no valid $4$-role coloring of $G'$ exists.} 

\noindent\textit{Proof} By Claim~1, $\alpha'(Q)\cap \alpha'(S)=\emptyset$. Therefore, there is no color that can be assigned to $S$ in a valid $4$-role coloring of $G'$.  \hfill$\qed$\\ 

\hfill $\qed$
\end{proof}
 

\begin{mylemma} 
 $k$-{\sc role coloring} is \npc on connected bipartite graphs, for any fixed $k \geq 5$. 
\end{mylemma}\label{lem:>5rcol} 

\begin{proof}

We show a reduction from {\sc Hypergraph-2-coloring}. Let $H=(Q_H, S_H)$ be an instance of {\sc Hypergraph-2-coloring} such that $H$ is $3$-regular. Let $G= (Q\cup S, E)$ be its canonical incidence graph. We create the instance $G''$ of {\sc $k$-Role coloring} by attaching to each $s \in S$ a ``\emph{pendant path}" $p^s$ = $\{p^s_1, ..., p^s_{k-3}\}$. Each path $p^s$ is attached to the vertex $s$ through the edge ($p^s_{k-3}, s)$. The set of all the pendant paths is denoted by $P^S$. Refer to Figure \ref{k>5pic}.

\begin{figure}[!ht]
\begin{center}
\includegraphics[width=0.5\textwidth]{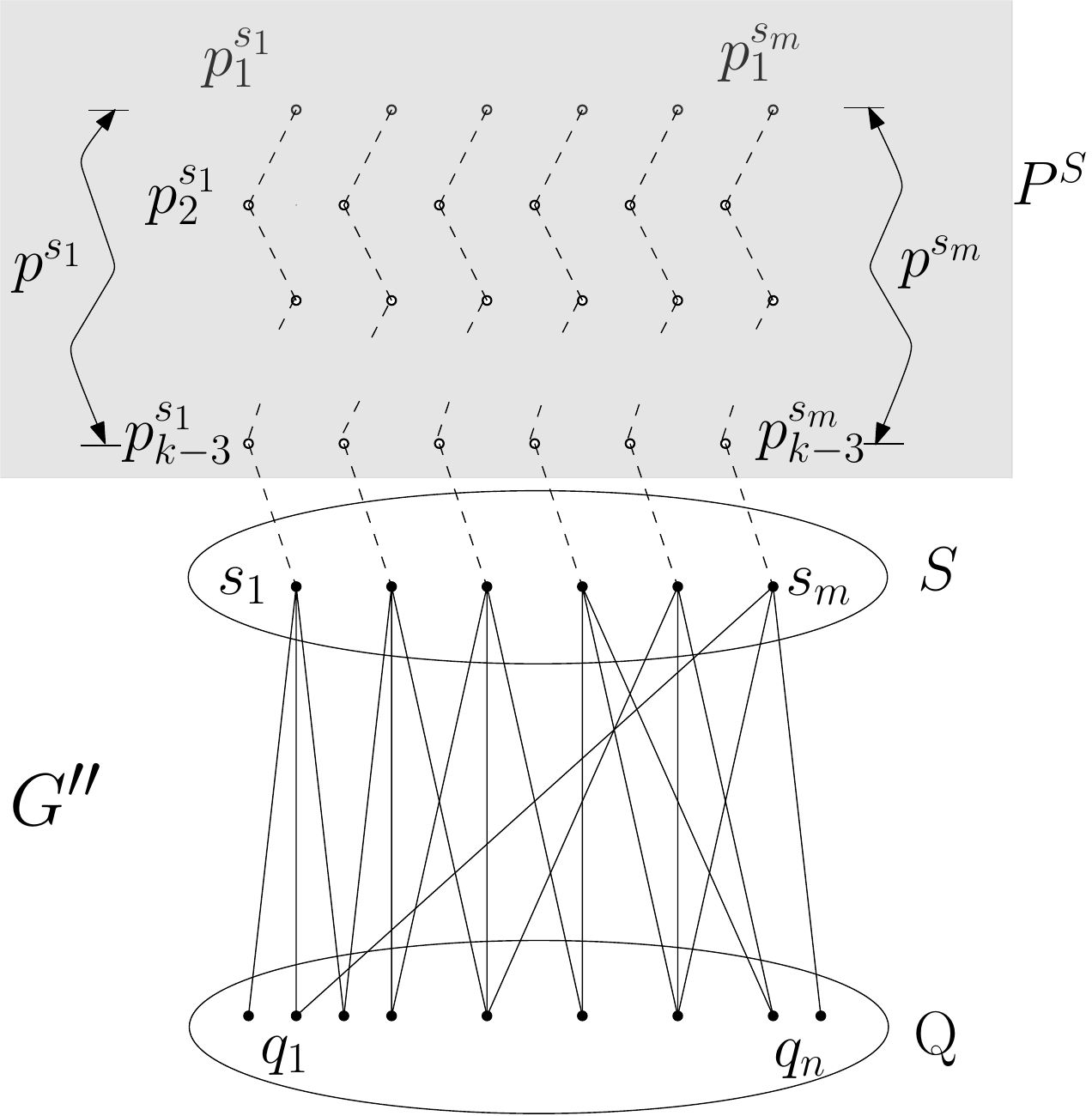}
\caption{We show the graph $G''$ in this figure which is a connected bipartite graph. The set $Q$ contains a vertex $q$ for every vertex $q \in Q_H$. The set $S$ contains a vertex $s$ for each hyperedge $s \in S_H$. To each $s \in S$ we attach a path of length $k-3$ denoted by $p^s$.}
\label{k>5pic}
\end{center}
\end{figure}

We prove that $H$ has a $2$-coloring if and only if $G''$ admits a $k$-role coloring. Suppose that the hypergraph $H$ has a $2$-coloring, $\beta$ : $Q_H \longrightarrow \{1, 2\}$. Let the role coloring for $G''$ be the following: $\alpha(q)$ =$\beta(q)$, for each $q \in Q$;  $\alpha(s)$= $3$, for all $s \in S$; $\alpha({p^{s}}_{k-3})= 4$, for all $s \in S$; $\alpha({p^{s}_j})$= $k-j+1$, where $ 1 \leq j \leq k-4$.
This gives us a $k$-role coloring of $G''$ with the corresponding role-graph $R$ defined as: $V(R)=\{1, 2, \dots, k\}$ and  $E(R) = \{ (1,3), (2,3)\} \cup\{(j, j+1) \mid  3 \leq j \leq k-1\}$.

We  now prove the converse. We assume that $G''$ has a valid $k$-role coloring $\alpha'$.

\paragraph{Claim~1:}
{\em In any valid $k$-role coloring of $G''$, each vertex of any pendant path $p^s$ gets a distinct color.}

\noindent\textit{Proof.} Assume the contrary. Let $p^s_i$ be the first vertex such that its neighbour $p^s_{i+1}$ gets the same color. For convenience, let us assume that $\alpha'(p^s_j) = j$ for $1 \leq j\leq i$. Then, $\alpha'(p^s_{i+2}) = i-1$ because any vertex colored $i$ needs a neighbour coloured $i-1$ and the only neighbour of $p^s_{i+1}$ other than $p^s_{i}$ is $p^s_{i+2}$. Similarly, $\alpha'(p^{i+3})= i-2$, and so on. Therefore the role graph $R'$ corresponding to $\alpha'$ would be disconnected with a component isomorphic to a $P_i$ with a loop on an end point. This is not possible as $G''$ is connected. \hfill $\qed$\\

\paragraph{Claim~2:}
{\em $\alpha'(Q) \cap \alpha'(S \cup P^S) =\emptyset $}.

\noindent\textit{Proof.} Suppose that there exists a vertex $s \in S$ such that $\alpha'(s) = \alpha'(q) = x$ for some $q\in Q$ and some $x \in \{1, 2, \ldots k\}.$ Consider its pendant path $p^s$. By the previous claim, there must be $k-3$ distinct colors given to the vertices of $p^s$. Without loss of generality, let these colors be $1, \ldots k-3$ such that $\alpha'(p^s_j)=j$ for all $1\leq j\leq k-3$. Therefore, $x$ must be adjacent to $k-3$ in $R'$. If $x = k-3$ then $s$ must have a neighbour $q'\in Q$ colored $k-4$. $q'$ in turn must be adjacent to some $s'\in S$ colored $k-5$. Consider the pendant path $p^{s'}$. Since $k-5$ is only adjacent to either $k-4$ or $k-6$, $p^{s'}_{k-3}$ must get one of the two colors. In the former case $\alpha'(p^{s'}_{2})=1$ and since $1$ is adjacent only to $2$, there is no color that can be given to $p^{s'}_{1}$. In the latter case, $\alpha'(p^{s'}_3)=1, \alpha'(p^{s'}_2) = 2$ but there is no valid color that can be given to $p^{s'}_1$ as $p^{s'}_2$ needs a neighbour colored $3$ but $deg_{R'}(3)>1$. 

Therefore, let $x = k-2$. Since $\alpha'(q)= k-2 $, $q$ needs at least one neighbour $s'\in S$ colored $k-3$ as $x$ is adjacent to $k-3$ in $R'$. The pendant path attached to $s'$ gets the colors $k-4, \ldots 1$ in order such that $\alpha'(p^{s'}_{k-3})=k-4$ and $\alpha'(p^{s'}_2)=1$. There is no valid color that can be given to $p^{s'}_1$. 

Similarly if there exists a vertex $q \in Q$ such that $\alpha'(q) \in \{1, 2, \ldots k-3\}$, then its neighbour $s \in S$ must get a color $i \in \{1, 2, \ldots k-3\}$ such that $i= \alpha'(q)+1$ or $i=\alpha'(q)-1$. As above, there would be no valid colour that can be given to $p^s_1$. This is a contradiction to the validity of $\alpha'$. \hfill$\qed$\\

\paragraph{Claim~3:}
{\em In any valid $k$-role coloring of $G''$, $Q$ cannot be monochromatic.}

\noindent\textit{Proof.} Assume the contrary. Without loss of generality, let $\alpha'(Q) = k$. By Claim~2, $\alpha'(Q) \cap \alpha'(S \cup P^S)=\emptyset$. Therefore, let $s \in S$ get the color $k-1$. We shall show that $S$ is monochromatic and $\alpha'(S)=k-1$. None of the vertices in $p^s$ could get the color $k$ by claim~2. So, $p^s_{k-3}$ could either get $k-1$ or a distinct color $k-2$. By claim~1, each vertex in $p^s$ gets a distinct color. Let the set of colors given to $p^s$ be $\{k-1, k-2, \ldots 3\} $ in the former case and $\{k-2, \ldots 2\}$ in the latter. Suppose that $q$ has another neighbour $s_1 \in S$ colored $1$. Then, $p^{s_1}_{k-3}$ could get color $2$ in the former case, but there would be no valid color that could be given to $p^{s_1}_{k-4}$ as none of the other colors are adjacent to $2$. In the latter case, there would be no valid color that could be given to $p^{s_1}_{k-3}$. Therefore, $S$ must be monochromatic with $k-1$. $P^S$ gets additional $k-3$ colors and the total number of colors assigned are $k-1$. This is a contradiction to $\alpha'$ being a $k$-role coloring. \hfill $\qed$\\ 

\hfill $\qed$

\end{proof} 

\section{3-role coloring bipartite chain graphs} \label{sec:poly}
   
\textbf{ Definition}\cite{chain,bichain} {\it A bipartite graph is a {\em chain graph} if and only if the neighborhood of the vertices of each partition can be linearly ordered with respect to inclusion. Equivalently, a bipartite graph is a chain graph if and only if it is $2K_2$-free.}
 
A {\em universal vertex}\footnotemark[3] is a vertex in one partition of a bipartite graph that is adjacent to every vertex in the other partition.
 
\footnotetext[3]{Our definition of ``{\em universal vertex}" is different from the standard definition.}
 
Throughout this section, we assume $G=  (X, Y, E)$ is a bipartite chain graph. $U_X$ and $U_Y$ are the sets of universal vertices in the partitions $X$ and $Y$, respectively. Likewise, $P_X$ and $P_Y$ are sets of pendant vertices in the partitions $X$ and $Y$. Note that $P_X$ and/or $P_Y$ may be empty. We also assume $|V(G)| \geq 3$ or else, $G$ is trivially a \no instance. In the rest of this section, we aim to prove the following theorem. 

\begin{mytheorem}
A bipartite chain graph $G=(X,Y,E)$, where $|X \cup Y| \geq 3$, has a valid $3$-role coloring if and only if any of the following is true:

\begin{enumerate}[label=\arabic*., font=\normalfont]
	\item \label{i. disconnected} $G$ is not connected
	\item \label{i.x=1} $|X|=1$ or $|Y|=1$
	\item \label{i.x=2,ux=2} $|X|=2$ (or $|Y|=2$) and both the vertices of $X$ (or $Y$) are universal vertices.
	\item \label{i.x=2,y>2} $|X|=2$, $|Y|>2$ and $|Y \setminus P_Y|>1 $. (symmetrically, $|Y|=2$, $|X|>2$ and $|X\setminus P_X|>1$)
	\item \label{i.x>2,y>2} $|X| \geq 3$ and $|Y| \geq 3$. 
\end{enumerate}

Furthermore, given any bipartite chain graph we can check in polynomial time whether it has a valid $3$-role coloring.
\end{mytheorem}\label{thm:bichain}
\begin{proof}
	We shall first prove that in each of the cases above, $G$ has a valid $3$-role coloring. Thereafter, we show that $G$ does not admit a valid $3$-role coloring in any other case.
	
\begin{enumerate}[label=\arabic*., font=\normalfont]
	\item We prove that there exists a valid $3$-role coloring of $G$, when it is disconnected and $|G|\geq 3$. Note that edges are present in at most one of the connected components of a disconnected bipartite chain graph, otherwise it would contain a $2K_2$. So there are two cases, first when $G$ is edgeless. As we have assumed that there are at least three vertices, thus the graph can be trivially $3$-role colored with the corresponding role graph being an edgeless graph on $3$ vertices. In the second case, exactly one component has edges. We color the connected component with edges using colors $1$ and $2$. This can be done as the graph is bipartite. Color the rest of the isolated vertices with color $3$. Note that this is a valid $3$-role coloring with role graph $R$ such that $V(R)=\{1,2,3\}$ and $E(R)= \{(1,2)\}$.
	
	\item Henceforth, we assume that $G$ is connected. As there are no isolated vertices in $G$, the linear ordering on the neighborhood of vertices in each partition ensures that the largest neighborhood contains all the vertices in the partition. Hence, each partition has at least one universal vertex. We claim that $G$ has a valid $3$- role coloring if $|X|=1$. The only vertex $v \in X$ is universal. Let $\alpha(v)=1$ and for a vertex $w \in Y$ $\alpha(w)=2$. Assign color $3$ to the rest of vertices in $Y$. It is easy to check that this coloring agrees with the role graph $R=(\{1,2,3\}, \{(1,2), (1,3)\})$. We can argue symmetrically for $Y$. 
	
	\item Assume that $|X|=2$ and $X=\{u,v\}$. We show that $G$ has a valid $3$-role coloring if both the vertices in $X$ are universal. Assign color $1$ to $u$, color $2$ to $v$ and color $3$ to vertices of $Y$. The corresponding role graph is defined as $R=(\{1,2,3\}, \{(1,3),(2,3)\})$.
	
	\item  Suppose that $|X|=2$ and $|Y|>2$, such that there exist at least one vertex in $Y$ that is not a pendant vertex. As before, let $X=\{u,v\}$ such that $u$ is universal, bur $v$ is not. First, let $P_X = \emptyset = P_Y$. Then, we define the role coloring assignment as $\alpha(u)= 1$, $\alpha(Y)= \{2\}$; and $\alpha(v)= 3$. Each vertex in $Y$ has a neighbor colored $1$, as $u$ is universal. Since there are no pendants in $G$, each of them has at least one other neighbor in $X$. This ensures that each vertex colored $2$ also has a neighbor colored $3$. 
	
	Next, we assume that $P_Y$ is non-empty. Let $Y= P_Y \cup T$ where $T$ is the set of vertices of degree two. We prove that $G$ has a valid $3$-role coloring if $|T| \geq 2$. Let $\alpha(a)= 1$, for some $a \in P_Y$. We cannot assign $1$ to $u$, due to the connectivity constraint of the role graph, by Observation~2.2. So, let $\alpha(u)= 2$. So $(1,2) \in E(R)$. Let $t\in T$ be a neighbor of $v$. By assumption \ref{i.x=2,y>2}, $|T| \geq 2$. We define the role coloring as: $\alpha(v) = \alpha(u) = 2$, $\alpha(t) = 1$. Also, color all the vertices in $P_Y$ with the color $1$ and all the remaining vertices in $T \setminus \{t\}$ with the color $3$. It is easy to verify that this coloring is consistent with the role graph $R=(\{1,2,3\}, \{(1,2),(2,3)\})$.
	
	\item We finally show that $G$ has a valid $3$-role coloring if it is connected and  $|X| \geq 3$, and $|Y| \geq 3$. We demonstrate a $3$-role coloring for each of the following sub-cases:\\	
	\paragraph{ Sub-case~1.} $P_X=P_Y=\emptyset$\\	
	Let $\alpha(u)= \{1\}$ for some $u\in U_X$, $\alpha(Y)=\{2\}$, and $\alpha(w)=3$, for all $w \in G \setminus N[U_X]$. The validity follows from the preceding case \ref{i.x=2,y>2}.\\	
	\paragraph{ Sub-case~2.} $P_X \neq \emptyset$, and $P_Y= \emptyset$\\
	The unique neighbor of a pendant in $P_X$ must be a universal vertex. Moreover, $|U_Y| = 1$.
	Suppose that $|U_Y| \geq 2$. Then every vertex in $X$ would have degree at least 2. This contradicts our assumption that $P_X \neq \emptyset$.
	
	Now, we define the role coloring of vertices in $G$ as follows: $\alpha:V(G) \longrightarrow \{1, 2, 3\}$, $\alpha(y) = 1$,  for all $y$ in $Y$; $\alpha(u) = 2$, where $u \in U_X$ and $\alpha(X \setminus \{u\})=3$. Every vertex colored $1$ is adjacent to the vertex colored $2$ (that is $u$) and at least one vertex colored $3$. Thus, $\alpha$ is a valid role-coloring of $G$.\\	
	\paragraph{ Sub-case~3.} Pendants occur in both the partitions.\\	
	Note that $|U_X|= |U_Y|= 1$. 	
	Let $x$ and $y$ be the universal vertices in $X$ and $Y$, respectively. We first assume that $N(x)  \cup   N(y)$ does not induce an independent set. Let $alpha:V \longrightarrow \{1,2,3\}$ be a role coloring of $G$, as before.
	
	In this case, $\alpha (P_X \cup P_Y) = \{1\}$; $\alpha(x)= \alpha(y)= 2$; $\alpha (G\setminus \{\{x\}\cup \{y\} \cup P_X \cup P_Y\})= \{3\}$. Therefore, every vertex colored $1$ is adjacent only to vertices colored $2$. The two vertices colored $2$ are adjacent to all the vertices in each partition and hence adjacent to both $1$ and $3$. Since the neighborhood of $x$ and $y$ do not form an independent set, at least one vertex in each partition is not a pendant. Hence, there are vertices colored $3$ which are adjacent to each other and the universal vertices and $\alpha$ is a valid role coloring of $G$.  
	
	Now, suppose that $N(x) \cup N(y)$ is an independent set. Then, $|N(x)| \geq 2$ and $|N(y)| \geq 2$, as each partition has at least $3$ vertices. Define the role coloring as: $\alpha(x)=\alpha(y) =2$. Color any one neighbor each of $x$ and $y$ with the color $1$ and another neighbor with the color $3$. The remaining vertices can be arbitrarily given the colors $1$ or $3$. $\alpha$ is a valid role coloring of $G$, in this case too.	
	
	
\end{enumerate}

We claim that if $G$ has a valid $3$-role coloring, it must satisfy at least one of the conditions in Theorem~\ref{thm:bichain}.1. Towards a contradiction suppose that none of the cases hold. Then, $|X|=2$, one of its vertices is not universal, $|Y|=2$ and $|Y \setminus P_Y|=1$. We also assume that $X=\{u,v\}$, and that $u$ is universal, whereas $v$ is not. This implies that $G$ is isomorphic to a $P_4$, which is a \no instance for $3$-\textsc{Role coloring}. 

To see that, assign the color $1$ to $a \in P_Y$. By Observation~2.2, the role graph $R$ corresponding to $\alpha$ must be connected. Thus, the neighbor of $u$ must get a new color, say $2$. If possible, let $\alpha(v) = 1$. It forces $\alpha(t) = 2$. This contradicts that $\alpha$ is a $3$-role coloring of $G$. So let $\alpha(v) = 2$. If $\alpha(t)=2$, it forces the role graph to have a single vertex labeled $2$, with a self loop which contradicts the fact that $(1,2) \in E(R)$. If $\alpha(t)=3$, then $v$ has no neighbor colored $1$, which is again a contradiction as $(1,2) \in E(R)$. Therefore, let $\alpha(v)=3$. Then, $\alpha(t) \neq 1$, as $t$ is adjacent to vertices colored $2$ and $3$, but the degree of the color $1$ in $R$ is equal to one. Also, $\alpha(t) \neq 2$, as $t$ has no neighbor colored $1$. If possible, let $\alpha(t)=3$. This is not a valid role coloring as $\alpha(v)=\alpha(t)=3$ but $\alpha(N(v))= \{1, 3\} \neq \{2, 3\}= \alpha (N(t))$. Hence if $G$ is isomorphic to a $P_4$, $G$ is a \textit{no} instance of $3$-\textsc{Role coloring}.\\ 

There are a finite number of cases for which $G$ can be $3$-role colored, and each case can be verified in $\mathcal{O}(n+m)$ time, where $n$ and $m$ are the number of vertices and edges of $G$. Since we have shown that a valid $3$-role coloring exists of these cases alone, it follows that each of them is a \yes instance of $3$-\textsc{Role coloring} and every other case is a \no instance. Hence, we can decide $3$-{\sc role coloring} on bipartite chain graphs in $\mathcal{O}(n+m)$.\hfill $\qed$
\end{proof}

\section{Almost bipartite graphs} \label{sec:para}

An ``{\em almost bipartite}" graph is a graph that has a set of at most $d$ vertices or $d$ edges, $d$ being a constant, which on deletion yield a bipartite graph. Here, we prove that $2$-{\sc Role coloring} is \npc on the class of almost bipartite graphs, even if $d=1$.

\begin{mytheorem}\label{thm:6}
$2$-{\sc Role coloring} is \npc on the class of almost bipartite graphs.
\end{mytheorem}
\begin{proof}
We know from \cite{Rrcol} that given a connected bipartite graph $G$ and a role graph $R$ isomorphic to an edge with a self loop incident on one of its end points, deciding if $G$ has a valid $R$-role coloring, is \npc. Thus, we show a reduction from this problem to $2$-{\sc Role coloring} on almost bipartite graphs with $d= 1$. Refer to Figure \ref{thm6pic}.

\begin{figure}[!ht]
\begin{center}
\includegraphics[width=0.7\textwidth]{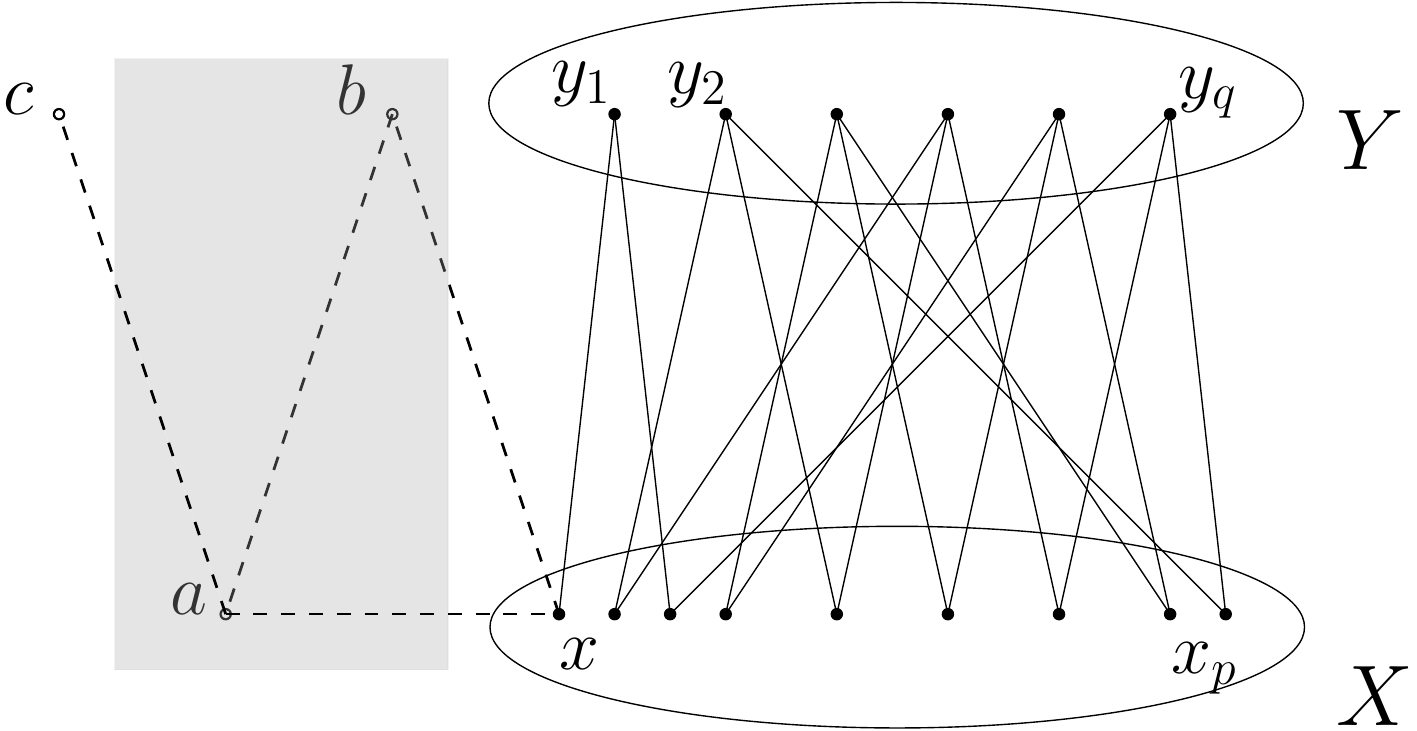}
\caption{The figure depicts a connected bipartite graph with the bipartition $X$ and $Y$. To the bipartite graph the vertices $a, b, \, \text{and}, c$ are added. Deleting the vertex $b$ yields a bipartite graph back.}
\label{thm6pic}
\end{center}
\end{figure}

Let $R$ be defined as $V(R)=\{1, 2\}$ and 
$E(R) = \{(1,1), (1,2)\}$. Let $G$ be a connected bipartite graph and let $x \in  V(G)$. We construct $G'$ from $G$ by adding to it 3 vertices $a,b$ and $c$ and the edges $(a,b), (a,c), (x, a),  (x, b)$. $G'$ is not bipartite as it contains a triangle, namely, $abx$. It is 1 vertex (remove $b$) and 1 edge (remove $(x,b)$) away from the bipartite graph $G$.

Suppose that $G$ is a \yes instance of $R$-{\sc Role coloring}. Then, it has a role coloring $\alpha$, with the corresponding role graph being isomorphic to $R$. Then we can extend $\alpha$ to the vertices of $G'$ as follows: If $\alpha(x) = 1$, then color $a$ with $1$ and $b$ and $c$  with $2$. If $\alpha(x) = 2$ then color both $a$ and $b$ with $1$ and $c$ with $2$. Thus, $G'$ has an $R$-role coloring and therefore a valid $2$-role coloring.

Conversely, suppose that $G'$ has a $2$-role coloring $\alpha'$. First assume that $\alpha'(a)=1$. If both $x$ and $b$ received the color $1$, then $b$ would only be adjacent to vertices colored $1$. Thus, the role graph would be disconnected, as $1$ could only be adjacent to itself. This is not possible as $G'$ is a connected graph. Hence, at least one of them gets the color $2$. Also, $c$ gets color $2$ as its degree is one, but the degree of $1$ in $R$ is two. Now, there are two possible subcases:\\
\noindent {\bf Subcase 1:} {\em $\alpha'(b)=1$ and $\alpha'(x)=2$.}\\
Any other neighbor $y$ of $x$ in $G'$ must be colored $1$. As per the assignment $\alpha'$, $y$ must have a neighbor colored $1$ and another neighbor colored $2$. Therefore, $\alpha'$ restricted to $G$ is an $R$-role coloring, and $G$ has a valid $R$-role coloring.\\
\noindent {\bf Subcase 2:} {\em $\alpha'(b)= 2$ and $\alpha'(x)=1$}\\
Suppose that all the other neighbors of $x$ in $G'$ are colored $2$. Those neighbors would, in turn, be adjacent to vertices colored $1$ in the partition of $x$. These vertices colored $1$ must obey the coloring $\alpha'$, and have neighbors of both the colors $1$ and $2$. Once again, restricting $\alpha'$ to $G$, gives an $R$-role coloring of $G$, showing that $G$ has a valid $R$-role coloring.

Symmetrically, if $\alpha'(a)= 2$, the role graph is again isomorphic to $R$. Hence, if $G'$ has a $2$-role coloring it has an $R$-role coloring and therefore $G$ has a valid $R$-role coloring. \qed
\end{proof}

\section{Conclusions}

We have proved that $k$-{\sc Role coloring}, for a fixed $k \geq 3$, remains \npc even when restricted to the class of bipartite graphs. We have also shown that $2$-{\sc Role coloring}, which is polynomial time solvable on bipartite graphs, becomes \npc on graphs that are a constant number of vertices or edges away from bipartite graphs. We do not know of any existing literature on \textsc{$k$-Role coloring} on graphs close to a specific hereditary class of graphs. This opens avenues for further research.

On a subclass of $2K_2$-free graphs, i.e, bipartite chain graphs, we have shown that $3$-{\sc Role coloring} can be solved in polynomial time. However, the complexity of $k$-{\sc Role coloring}, for $k$= $2$ and $3$, is open on general $2K_2$-free graphs.

\subsection*{Acknowledgements}
The authors would like to thank Professor Venkatesh Raman for several fruitful discussions and his invaluable feedback. We are also grateful to the anonymous reviewers whose thorough and precise feedback helped us improve the paper.

\bibliographystyle{splncs04}

\bibliography{citations}

\end{document}